\documentclass[
final
]{dmtcs-episciences}

\usepackage[utf8]{inputenc}
\usepackage{subfigure}

\usepackage{complexity}
\usepackage{amsmath}
\usepackage{amsfonts}
\usepackage{amssymb}
\usepackage{mathtools}
\usepackage{algorithm}
\usepackage{algpseudocode}
\usepackage{xspace}

\usepackage[amsthm,thmmarks]{ntheorem}
\usepackage[numbers]{natbib}

\DeclareMathOperator{\diam}{diam}
\DeclareMathOperator{\rc}{rc}
\DeclareMathOperator{\src}{src}

\newcommand{\ProblemFormat}[1]{\textsc{#1}}
\newcommand{\ProblemName}[1]{\ProblemFormat{#1}\xspace}

\newcommand{\probkSRC}[0]{\ProblemName{$k$-SRC}}
\newcommand{\probkSubsetSRC}[0]{\ProblemName{$k$-SSRC}}

\newtheorem{theorem}{Theorem}
\newtheorem{lemma}[theorem]{Lemma}
\newtheorem{corollary}[theorem]{Corollary}
\newtheorem{conjecture}[theorem]{Conjecture}

\newcommand{\heading}[1]{\medskip\noindent{\bf #1.\ }}%

\author{Melissa Keranen\affiliationmark{1}
  \and Juho Lauri\affiliationmark{2}}
  \title{Computing minimum rainbow and strong rainbow colorings of block graphs}

\affiliation{
  Michigan Technological University, USA\\
  Tampere University of Technology, Finland. Current affiliation: Nokia Bell Labs, Ireland}
\received{2017-8-23}
\revised{2018-4-2}
\accepted{2018-5-25}
\begin{document}
%\publicationdetails{VOL}{2015}{ISS}{NUM}{SUBM}
\publicationdetails{20}{2018}{1}{22}{3877}
\maketitle
\begin{abstract}
  A path in an edge-colored graph $G$ is \emph{rainbow} if no two edges of it are colored the same. The graph $G$ is \emph{rainbow-connected} if there is a rainbow path between every pair of vertices. If there is a rainbow shortest path between every pair of vertices, the graph $G$ is \emph{strongly rainbow-connected}.
The minimum number of colors needed to make~$G$ rainbow-connected is known as the \emph{rainbow connection number} of~$G$, and is denoted by $\rc(G)$. Similarly, the minimum number of colors needed to make~$G$ strongly rainbow-connected is known as the \emph{strong rainbow connection number} of~$G$, and is denoted by~$\src(G)$. 
We prove that for every~$k \geq 3$, deciding whether $\src(G) \leq k$ is $\NP$-complete for split graphs, which form a subclass of chordal graphs. Furthermore, there exists no polynomial-time algorithm for approximating the strong rainbow connection number of an $n$-vertex split graph with a factor of $n^{1/2-\epsilon}$ for any~$\epsilon > 0$ unless~$\P = \NP$. 
We then turn our attention to block graphs, which also form a subclass of chordal graphs.
We determine the strong rainbow connection number of block graphs, and show it can be computed in linear time.
Finally, we provide a polynomial-time characterization of bridgeless block graphs with rainbow connection number at most 4.
\end{abstract}

\section{Introduction}
Let $G$ be an edge-colored undirected graph that is simple and finite. A path in $G$ is \emph{rainbow} if no two edges of it are colored the same. The graph $G$ is \emph{rainbow-connected} if there is a rainbow path between every pair of vertices. If there is a rainbow shortest path between every pair of vertices, the graph $G$ is \emph{strongly rainbow-connected}. The minimum number of colors needed to make $G$ rainbow-connected is known as the \emph{rainbow connection number} of $G$ and is denoted by $\rc(G)$. Likewise, the minimum number of colors needed to make $G$ strongly rainbow-connected is known as the \emph{strong rainbow connection number} of $G$ and is denoted by $\src(G)$. Rainbow connectivity was introduced by Chartrand, Johns, McKeon, and Zhang~\cite{Chartrand2008} in 2008. While being a theoretically interesting way of strengthening connectivity, rainbow connectivity also has possible applications in data transfer and networking~\cite{Li2012}. The study of rainbow colorings and several of its variants have recently attracted increasing attention in the research community. For a comprehensive treatment, we refer the reader to the books~\cite{Chartrand2008b, Li2012b}, or the recent survey~\cite{Li2012}.

Denote by $n$ the number of vertices and by $m$ the number of edges of a graph in question. It is easy to verify that $\rc(G) \leq n - 1$; indeed, such an edge-coloring is obtained by coloring the edges of a spanning tree of $G$ in distinct colors. On the other hand, we will always need as at least as many colors as is the length of a longest shortest path in $G$. Thus, an easy lower bound for $\rc(G)$ is given by the diameter of $G$, denoted by $\diam(G)$. That is, we have that $\diam(G) \leq \rc(G) \leq n-1$. It also holds that $\rc(G) \leq \src(G)$, since every strongly rainbow-connected graph is also rainbow-connected.
For extremal cases, it is easy to see that $\rc(G) = \src(G) = 1$ if and only if $G$ is a complete graph. Similarly, we have that $\rc(G) = \src(G) = m$ if and only if $G$ is a tree (for proofs, see~\cite{Chartrand2008}). Chartrand~\emph{et al.}~\cite{Chartrand2008} also determined the exact rainbow and strong rainbow connection numbers for some structured graph classes, including cycles, wheel graphs, and complete multipartite graphs.

Not surprisingly, determining the rainbow connection numbers is computationally hard. Chakraborty, Fischer, Matsliah, and Yuster~\cite{Chakraborty2009} proved that given a graph $G$, it is $\NP$-complete to decide whether $\rc(G)=2$. Ananth, Nasre, and Sarpatwar~\cite{Ananth2011} further showed that for every $k \geq 3$, deciding whether $\rc(G) \leq k$ is $\NP$-complete. Using different ideas, a proof of hardness for every $k \geq 2$ is also given by Le and Tuza~\cite{Le:tech}. The hardness of computing the strong rainbow connection number was shown by Ananth {\em et al.}~\cite{Ananth2011} as well. In particular, they proved that for every $k \geq 3$, deciding whether $\src(G) \leq k$ is $\NP$-complete, even when $G$ is bipartite. Furthermore, as $\rc(G) = 2$ if and only if $\src(G) = 2$ (for a proof, see~\cite{Chartrand2008}), it follows that deciding whether $\src(G) \leq k$ is $\NP$-complete for every $k \geq 2$.

Because rainbow-connecting graphs optimally is hard in general, there has been interest in approximation algorithms and easier special cases. Basavaraju, Chandran, Rajendraprasad, and Ramaswamy~\cite{Basavaraju2012} presented approximation algorithms for computing the rainbow connection number with factors $(r+3)$ and $(d+3)$ respectively, where $r$ is the radius and $d$ the diameter of the input graph. Chandran and Rajendraprasad~\cite{Chandran2013} proved that there is no polynomial-time algorithm to rainbow-connect graphs with less than twice the optimum number of colors, unless $\P = \NP$. Ananth {\em et al.}~\cite{Ananth2011} showed that there is no polynomial time algorithm for approximating the strong rainbow connection number of an $n$-vertex graph with a factor of $n^{1/2-\epsilon}$, where $\epsilon > 0$ unless $\NP = \ZPP$.

There is a line of research studying rainbow connection on chordal graphs (see e.g.,~\cite{Basavaraju2012,Chandran2012,Chandran2013,Chandran2014}). 
In this regard, it is known to be $\NP$-complete to decide whether $\rc(G) \leq k$ for every $k \geq 2$ even when $G$ is chordal~\cite{Chandran2012,Chandran2014}.
Furthermore, the rainbow connection number of a chordal graph can not be approximated to a factor less than $5/4$ unless $\P = \NP$~\cite{Chandran2013}. 
Motivated by this result, there has been interest in a deeper investigation of the rainbow connection number of subclasses of chordal graphs.
Chandran, Rajendraprasad, and Tesa{\v{r}}~\cite{Chandran2014} showed that for split graphs, the problem of deciding whether $\rc(G) = k$ is $\NP$-complete for $k \in \{2,3\}$, and in $\P$ for all other values of $k$. 
Chandran and Rajendraprasad~\cite{Chandran2012} showed split graphs can be rainbow-connected in linear-time using at most one more color than the optimum. 
In the same paper, the authors also gave an exact linear-time algorithm for rainbow-connecting threshold graphs.
Furthermore, they noted that their result is apparently the first efficient algorithm for optimally rainbow-connecting any non-trivial subclass of graphs.
To the best of our knowledge, the complexity of strongly rainbow-connecting chordal graphs is an open question.
Moreover, we are not aware of any efficient exact algorithms for computing the strong rainbow connection number of a non-trivial subclass of graphs.

\heading{Our results}
We further investigate the rainbow and strong rainbow connection number of subclasses of chordal graphs.
We extend the known hardness results for computing the strong rainbow connection number by showing it is $\NP$-complete to decide whether a given split graph can be strongly rainbow-connected in $k$ colors, where $k \geq 3$. 
As a by-product of the proof, we obtain that there exists no polynomial-time algorithm for approximating the strong rainbow connection number of an $n$-vertex split graph with a factor of $n^{1/2-\epsilon}$ for any $\epsilon > 0$ unless $\P = \NP$. 
These negative results further motivate the investigation of tractable special cases. Indeed, we determine the strong rainbow connection number of block graphs, and show that any block graph can be strongly rainbow-connected optimally in linear time. 
Finally, we turn to the rainbow connection number, and characterize the bridgeless block graphs with rainbow connection number 2, 3, or 4.

\section{Preliminaries}
The graphs considered are connected, simple, and undirected. For graph-theoretic concepts not covered here, we refer the reader to~\citep{Diestel2005}. %We write $[n] = \{1,2,\ldots,n\}$ for a positive integer $n$.

Let $G=(V,E)$ be a graph. The \emph{diameter} of $G$, denoted by $\diam(G)$, is the length of a longest shortest path in $G$. The \emph{degree} of a vertex is the number of edges incident to it. The \emph{minimum degree} of $G$, denoted by $\delta(G)$, is the minimum of the degrees of all the vertices in $G$. If $G$ is obvious from the context, we may shorten $\delta(G)$ to $\delta$. Finally, a \emph{dominating set} is a subset $D \subseteq V$ of vertices such that every vertex in $V \setminus D$ is adjacent to at least one vertex in $D$. If $D$ induces a connected subgraph in $G$, we say $D$ is a \emph{connected dominating set}. The minimum size of a (connected) dominating set in $G$, denoted by ($\gamma_c(G)$) $\gamma(G)$, is known as the \emph{(connected) domination number} of $G$.

A \emph{chord} is an edge joining two non-consecutive vertices in a cycle. A graph is \emph{chordal} if every cycle of length 4 or more has a chord. Equivalently, a graph is chordal if it contains no induced cycle of length~4 or more. 
Let us introduce some subclasses of chordal graphs that are most central for this work. A \emph{split graph} is a graph whose vertex set can be partitioned into a clique and an independent set. A \emph{cut vertex} is a vertex whose removal will disconnect the graph. A \emph{biconnected graph} is a connected graph having no cut vertices. In a \emph{block graph}, every maximal biconnected component, known as a \emph{block}, is a clique. In a block graph $G$, different blocks intersect in at most one vertex, which is a cut vertex of $G$. In other words, every edge of $G$ lies in a unique block, and $G$ is the union of its blocks. A particular property of block graphs is that they are \emph{geodetic}, meaning there is exactly one shortest path between every pair of vertices (see e.g.,~\cite{Stemple1968}). Both split graphs and block graphs are chordal. 

The concept of separators is central to chordal graphs. A set $S \subseteq V$ disconnects a vertex $a$ from vertex $b$ in a graph $G$ if every path of $G$ between $a$ and $b$ contains a vertex from $S$. A non-empty set $S \subseteq V(G)$ is a \emph{minimal separator} of $G$ if there exists $a$ and $b$ such that $S$ disconnects $a$ from $b$ in $G$, and no proper subset of $S$ disconnects $a$ from $b$ in $G$. If we want to identify the vertices that $S$ disconnects, we may also refer to $S$ as a \emph{minimal $a$-$b$ separator}. 
For a more comprehensive treatment on chordal graphs, we refer the reader to~\citep{Blair1993}.

\section{Hardness of strongly rainbow-connecting split graphs}
\label{sec:hardness_src}
In this section, we show that deciding whether a split graph can be strongly rainbow-connected with $k \geq 3$ colors is $\NP$-complete. We remark that it follows from the work of Chandran~\emph{et al.}~\cite{Chandran2014} that the problem is $\NP$-complete for $k=2$; however, to the best of our knowledge, the complexity of the problem for $k \geq 3$ has been open even for chordal graphs.

In the \emph{$k$-subset strong rainbow connectivity problem} (\probkSubsetSRC), we are given a graph $G$, a set of pairs $P \subseteq V(G) \times V(G)$, and an integer $k$. The goal is to decide whether $E(G)$ can be colored with $k$ colors such that each pair of vertices in $P$ is connected by a rainbow shortest path. The problem was shown to be $\NP$-complete by Ananth {\em et al.}~\cite{Ananth2011} even when the graph $G$ is a star.
\begin{lemma}[\hspace{1sp}{\cite{Ananth2011}}]
For every $k \geq 3$, the \probkSubsetSRC problem is $\NP$-complete when the graph $G$ is a star.
\end{lemma}
We reduce from this problem, and make use of some ideas of~\cite{Chakraborty2009} in the following. For convenience, we denote by \probkSRC the problem of deciding whether a given a graph $G$ can be strongly rainbow-connected in $k$ colors.
\begin{theorem}
\label{thm_src_chordal_hardness}
For every integer $k \geq 3$, it is $\NP$-complete to decide if $\src(G) \leq k$, where $G$ is a split graph.
\end{theorem}
\begin{proof}
Let $I = (S, P, k)$ be an instance of the \probkSubsetSRC problem, where $S=(V,E)$ is a star, both $p$ and $q$ in each $(p,q) \in P$ are leaves of $S$, and $k \geq 3$ is an integer. We construct an instance $I' = (G')$ of \probkSRC, where $G'=(V',E')$ is a split graph such that $I$ is a YES-instance of \probkSubsetSRC iff $I'$ is a YES-instance of \probkSRC.

\begin{figure}[t]
\centering
\includegraphics[scale=1]{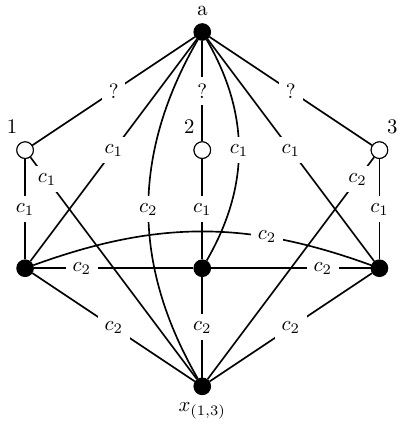}
\caption{A star graph $S$ on the vertex set $\{a,1,2,3\}$ transformed to a split graph $G'$ with $P = \{ (1,2),(2,3) \}$. The white vertices form an independent set while the black vertices form a clique. The symbol ?\ marks an edge-coloring $\chi$ of $S$ with $k$ colors under which the pairs in $P$ are connected by a rainbow path.}
\label{fig:reduction_chordal}
\end{figure}

Let $a$ be the central vertex of $S$. For every vertex $v \in V \setminus \{a\}$, we add a new vertex $x_v$, and for every pair of leaves $(u,v) \in (V \times V) \setminus P$, we add a new vertex $x_{(u,v)}$. Formally, we construct $G'=(V',E')$ such that
\begin{itemize}
\item $V' = V \cup \{ x_v \mid v \in V \setminus \{a\}\} \cup \{ x_{(u,v)} \mid (u,v) \in (V \times V) \setminus P\}$,
\item $E' = E \cup E_1 \cup E_2 \cup E_3$,
\item $E_1 = \{ (v,x_v), (a,x_v) \mid v \in V \setminus \{a\} \}$,
\item $E_2 = \{ (u,x_{(u,v)}), (v,x_{(u,v)}), (a,x_{(u,v)}) \mid (u,v) \in (V \times V) \setminus P \}$, and
\item $E_3 = \{ (x,x') \mid x,x' \in V' \setminus V \}$.
\end{itemize}
Let us then verify that $G'$ is a split graph. 
Observe the leaves of $S$ form an independent set in $G'$. 
The remaining vertices $\{a\} \cup (V' \setminus V)$ form a clique, proving $G'$ is split. Moreover, $a$ is a dominating vertex. 
An example illustrating the construction is given in Figure~\ref{fig:reduction_chordal}.

We will now prove $G'$ is strongly rainbow-connected with $k$ colors if and only if $(S, P)$ is $k$-subset strongly rainbow-connected. First, suppose $(S, P)$ is not $k$-subset strongly rainbow-connected; we will show $G'$ is not strongly rainbow-connected with $k$ colors. Observe that for each $(p,q) \in P$, there is a unique shortest path between~$p$ and~$q$ in~$S$. Moreover, the same holds for $G'$. Therefore, any strong rainbow coloring using $k$ colors must make this path strongly rainbow-connected in $G'$. But because the pairs in~$P$ cannot be strongly rainbow-connected with $k$ colors in~$S$, the graph~$G'$ cannot be strongly rainbow-connected with $k$ colors.

Finally, suppose $(S, P)$ is $k$-subset strongly rainbow-connected under some edge-coloring $\chi : E \to \{c_1,\ldots,c_k\}$. We will describe an edge-coloring $\chi'$ given to $G'$ by extending $\chi$. We retain the original coloring on the edges of $S$, that is, $\chi'(e) = \chi(e)$, for every $e \in E$. The rest of the edges are colored as follows:
\begin{itemize}
\item $\chi'(e) = c_1$, for all $e \in E_1$,
\item $\chi'(e) = c_2$, for all $e \in E_3$, and
\item $\chi'(ux_{(u,v)})=c_1$, $\chi'(vx_{(u,v)}) = c_2$, and $\chi'(ax_{(u,v)}) = c_2$ for all $(u,v) \notin P$.
\end{itemize}
It is straightforward to verify $G'$ is indeed strongly rainbow-connected under $\chi'$, completing the proof.~$\square$
\end{proof}

Ananth {\em et al.}~\cite{Ananth2011} reduced the problem of deciding whether a graph $G$ has chromatic number at most~$k$ to \probkSubsetSRC. 
Finally, they reduced \probkSubsetSRC to the problem of deciding whether a bipartite graph $G'$ can be strongly rainbow-connected in $k$ colors.
In this final step, the size of $G'$ is quadratic in the input graph~$G$ of the chromatic number instance.
Moreover, since the chromatic number of an $n$-vertex graph cannot be approximated with a factor of $n^{1-\epsilon}$ for any $\epsilon > 0$ unless $\P = \NP$~\cite{Zuckerman2006}, they obtained that $\src(G')$ cannot be approximated with a factor of $n^{1/2-\epsilon}$, under the same complexity-theoretic assumptions.
We apply precisely the same reasoning to the split graph~$G'$ obtained in Theorem~\ref{thm_src_chordal_hardness}, where $G'$ has size quadratic in~$G$ (similarly assuming a chain of reductions from an arbitrary instance $G$ of chromatic number).
We obtain the following.
\begin{theorem}
There is no polynomial-time algorithm that approximates the strong rainbow connection number of an $n$-vertex split graph with a factor of $n^{1/2-\epsilon}$ for any $\epsilon > 0$, unless $\P = \NP$.
\end{theorem}

\section{Strongly rainbow-connecting block graphs in linear time}
\label{sec:algorithm}
In this section, we determine exactly the strong rainbow connection number of block graphs.\footnote{We remark that the presentation given here is simpler than the one given in the doctoral thesis of the second author~\cite[Section~5.2]{Lauri-phd}.} Furthermore, we present an exact linear-time algorithm for constructing a strong rainbow coloring using $\src(G)$ colors for a given block graph $G$. If an explicit coloring is not required (i.e., if the value of $\src(G)$ suffices), the algorithm can be further simplified.

Let $B$ be a block in a block graph $G$ whose edges are colored by using colors from the set $R = \{c_1,\ldots,c_r\}$. Then we say that $B$ is \emph{colored} and $B$ is \emph{associated} with each color $c_1,\ldots,c_r$. 
In particular, if $B$ is associated with a color $c$ and no other block is associated with $c$, then we say $B$ is \emph{uniquely associated} with color $c$.
Furthermore, any color from $R$ can be used as a representative for the color of $C$. Thus we may say that $B$ has been colored $c_i$ for any $i \in \{ 1,\ldots,r \}$.
\begin{lemma}
\label{lemma_coloring}
Let $G$ be a block graph, let $B$ be a block that is uniquely associated with color $c$, let $(u,v)$ be an edge in $G$ such that $u,v \notin B$, and let $y$ be the minimal $a$-$b$ separator for any $a \in B \setminus \{y\}$ and $b \in \{u,v\}$. If no shortest $y$-$u$ path or shortest $y$-$v$ path contains $(u,v)$, then by coloring $(u,v)$ with the color $c$, any shortest path between $u$ or $v$ and $w \in B$ contains at most one edge of color $c$.
\end{lemma}
\begin{proof}
Any shortest path between $u$ or $v$ and $y$ does not contain the edge $(u,v)$, and does not contain any edges in $B$, so these paths do not have any edges of color $c$. Any shortest path between $y$ and $w$ is just an edge of color $c$.~$\square$
\end{proof}

The algorithm for strongly rainbow-connecting a block graph is presented in Algorithm~\ref{alg:src_block}. 
Given a block graph $G$, the algorithm partitions the blocks of $G$ into two sets $\mathcal{V}_{<3}$ and $\mathcal{V}_{\geq 3}$ based on the number of cut vertices contained in each block.
That is, if a block contains less than~3 cut vertices, it is added to $\mathcal{V}_{<3}$.
Otherwise, it is added to $\mathcal{V}_{\geq 3}$. 
Then, for each block in $\mathcal{V}_{<3}$, we introduce a new color and use it to color the edges of the block.
At the final step the algorithm goes through every block $B \in \mathcal{V}_{\geq 3}$.
Denote by $C(B)$ the set of cut vertices in $B$.
Fix~3 distinct vertices $b_1$, $b_2$, and $b_3$ in $C(B)$. 
Observe that in $G \setminus E(B)$, we would have at least 3 connected components, and $b_1$, $b_2$, and $b_3$ would be in different connected components. 
Suppose $E(B)$ was removed, and from each connected component $b_1$, $b_2$, and $b_3$ is in, pick a block in $\mathcal{V}_{<3}$. 
The picked three blocks are each associated with a distinct color. 
These colors are then used to color the edges of the block $B$.
The algorithm is illustrated in Figure~\ref{fig:src-alg-ex}: note that there are several choices of how the blocks in $\mathcal{V}_{\geq 3}$ are colored in the example, and the illustration shows one possibility.

\begin{figure}[t]
  \begin{center}
    \subfigure[]{\includegraphics[width=0.475\textwidth,keepaspectratio]{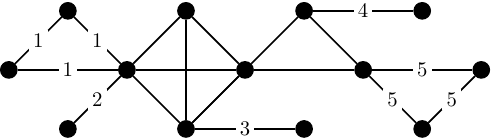}}
    \hfill
    \subfigure[]{\includegraphics[width=0.475\textwidth,keepaspectratio]{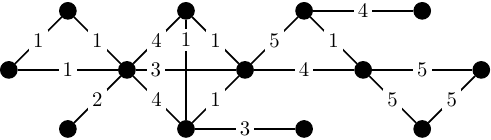}}
    \caption{\textbf{(a)} Lines 1 to 5 of Algorithm~\ref{alg:src_block} have been executed on an input block graph. \textbf{(b)} A strong rainbow coloring of~$G$ obtained after the execution of Algorithm~\ref{alg:src_block}.}
    \label{fig:src-alg-ex}
  \end{center}
\end{figure}

\begin{algorithm}[t]
  
  \caption{Algorithm for strong rainbow coloring a block graph}
  \label{alg:src_block}

  \begin{algorithmic}[1]
  \Require A block graph $G$
  \Ensure A strong rainbow coloring of $G$
  
    \State $\mathcal{V}_{<3} := \{ U \mid U \in \mathcal{B}(G) \wedge |C(U)| < 3 \}$ \Comment{Denote by $\mathcal{B}(G)$ the blocks of $G$}
    \State $\mathcal{V}_{\geq 3} := \mathcal{B}(G) \setminus \mathcal{V}_{<3}$
    
    \ForAll{$U \in \mathcal{V}_{<3}$}
    	\State Color edges in $U$ with a fresh distinct color
    \EndFor
    
   \ForAll{$B \in \mathcal{V}_{\geq 3}$}
    \State Let $b_1$, $b_2$, $b_3$ be distinct cut vertices in $C(B)$
    \State Assume $E(B)$ was removed from $E(G)$
    \State From each connected component $C_1$, $C_2$, $C_3$ is in, find a block in $\mathcal{V}_{<3}$
    \State Let $c_1,c_2,c_3$ be the respective colors associated with the found blocks
    \State Color all edges not incident to $b_1$ with color $c_1$
	\State Color all edges incident to $b_1$, except $(b_1,b_2)$, with color $c_2$
	\State Color the edge $(b_1,b_2)$ with color $c_3$
   \EndFor
	\end{algorithmic}
\end{algorithm}

The correctness of Algorithm~\ref{alg:src_block} is established by an invariant, which says that we always maintain the property that if the shortest path between two vertices is colored, then it is rainbow. We refer to this property as the \emph{shortest rainbow path property}.

\begin{theorem}
\label{thm_alg_correctness}
At every step, Algorithm~\ref{alg:src_block} maintains the shortest rainbow path property. 
\end{theorem}
\begin{proof}
Before the execution of the first loop, nothing is colored so the claim is trivially true. Furthermore, the first loop obviously maintains the property. To see this, consider any shortest path of length $\ell \geq 1$ at any step. The path consists of $\ell$ edges that are in $\ell$ distinct blocks. Since each colored block has received a distinct color, the shortest path is rainbow. This establishes the base step for the correctness of the second loop.

Assume after iteration $i-1$ of the second loop, if the shortest path between any two vertices is colored, then it is rainbow. 
We show that this property is maintained after iteration $i$ of the second loop.
Consider any edge $(u,v)$ in $B$ not incident to $b_1$, and let $y \in C_1$ be the minimal $a$-$b$ separator for any $a \in C_1 \setminus \{y\}$ and $b \in \{u,v\}$.
The algorithm states that $(u,v)$ will be colored with color $c_1$, which is uniquely associated with a block in $C_1$.
Because $u$ and $v$ are both at a distance 1 from $b_1$, it follows that neither shortest path $y$-$u$ or $y$-$v$ contains $(u,v)$. Thus by Lemma~\ref{lemma_coloring}, if the shortest $w$-$u$ path, for $w \in C_1$ is colored, then it is rainbow. (The same is true for the shortest $w$-$v$ path). Therefore, by coloring $(u,v)$ with color $c_1$, the shortest rainbow path property is maintained.

Consider any edge $(u,v)$ in $B$ not incident to $b_2$, and let $y \in C_2$ be the minimal $a$-$b$ separator for any $a \in C_2 \setminus \{y\}$ and $b \in \{u,v\}$.
By Lemma~\ref{lemma_coloring}, this edge can be colored with color $c_2$ to maintain the shortest rainbow path property. 
Notice that $u$ and $v$ are both at a distance 1 from $b_2$, and because all edges not incident to $b_1$ have been colored, it follows that $b_1$ must be one of these vertices (i.e., either $u=b_1$ or $v=b_1$). 
So we conclude that every edge incident to $b_1$, except $(b_1,b_2)$, can be colored with $c_2$ to maintain the shortest rainbow path property.

Now the only uncolored edge in $C$ is the edge $(b_1,b_2)$. Because $b_1$ and $b_2$ are both at a distance 1 from $b_3$, Lemma~\ref{lemma_coloring} assures us that by coloring $(b_1,b_2)$ with color $c_3$, the shortest rainbow path property is maintained.~$\square$
\end{proof}

Let us then consider the complexity of Algorithm~\ref{alg:src_block}. 
It is an easy observation that lines 1 to 5 take linear time. 
Observe that on line 9, we essentially perform reachability queries of the form \emph{given a block $B \in \mathcal{V}_{\geq 3}$, return a block containing less than 3 cut vertices that is reachable from cut vertex $b \in B$ with a path containing no cut vertices of $B$ besides $b$}.
In our context, such a query is performed for each cut vertex $b_1$, $b_2$, and $b_3$.
The naive way of answering such queries is to start a depth-first search (DFS) from each $b_1$, $b_2$, and $b_3$, and halt when a suitable block is found. However, such implementation requires $\Omega(d)$ time, where $d$ is the diameter of the input graph $G$. Using elementary techniques, we can preprocess the block graph $G$ before the execution of line 1 using linear time to answer such queries in $O(1)$ time. Thus, the total runtime will be linear as the for-loop on line 6 loops $O(n)$ times.

\begin{theorem}
Algorithm~\ref{alg:src_block} constructs a strong rainbow coloring in $O(n+m)$ time.
\end{theorem}
As Algorithm~\ref{alg:src_block} is correct and uses $k$ colors where $k$ is the number of blocks containing less than 3 cut vertices, we establish that $\src(G) \leq k$.
In the following, we prove that this is in fact optimal by showing a matching lower bound.
\begin{lemma}
\label{lem:src_lower_bound}
Let $G$ be a block graph, and let $k$ be the number of blocks containing less than 3 cut vertices. Then, $\src(G) \geq k$.
\end{lemma}
\begin{proof}
Let $A$ be a set of $k$ edges in $G$, one from each block containing less than 3 cut vertices, selected as follows.
For each block $B \in \mathcal{B}(G)$, if $|C(B)| = 1$, pick an edge incident to the cut vertex.
On the other hand, if $|C(B)| = 2$, pick the edge connecting the two cut vertices.
We claim that if we are to strongly rainbow-connect $G$, then the edges in $A$ must all receive distinct colors.

Suppose there are 2 edges in $A$ that are of the same color, say $(u,x) \in E(B)$ and $(v,y) \in E(B')$. 
Without loss, we may assume that $u$ and $v$ are cut vertices of $B$ and $B'$, respectively, such that $d(u,v)$ is minimized.
As $G$ is geodetic the $x$-$y$ shortest path is unique, and it contains two edges of the same color.~$\square$
\end{proof}

\noindent By combining the previous lemma with Theorem~\ref{thm_alg_correctness}, we arrive at the following.
\begin{theorem}
\label{thm_src_equals_k}
Let $G$ be a block graph, and let $k$ be the number of blocks containing less than 3 cut vertices. Then, $\src(G) = k$.
\end{theorem}
If an explicit coloring is not required, then it is easy to see that there is a linear-time algorithm for computing $\src(G)$, where $G$ is a block graph. This is obtained by counting the number of blocks containing less than 3 cut vertices.
\begin{corollary}
There is an algorithm such that given a block graph $G$, it computes $\src(G)$ in $O(n+m)$ time.
\end{corollary}

\section{On the rainbow connection number of block graphs}
\label{sec:block_rc}
In this section, we consider the rainbow connection number of block graphs. As a main result of the section, we prove a polynomial-time characterization of bridgeless block graphs with rainbow connection number at most~4.

Using known results, we begin by observing a tight linear-time computable upper bound on the rainbow connection number of a block graph of minimum degree at least~2. The following result was obtained by Chandran~\emph{et al.}~\cite{Chandran2011}.
\begin{theorem}[\hspace{1sp}{\cite{Chandran2011}}]
For every connected graph $G$, with $\delta(G) \geq~2$,
\begin{equation*}
\rc(G) \leq \gamma_c(G) + 2.
\end{equation*}
\end{theorem}
Further, the connected domination number of block graphs has been determined by Chen and Xing~\cite{Chen2004}.
\begin{theorem}[\hspace{1sp}{\cite{Chen2004}}]
Let $G$ be a connected block graph, let $S$ be the set of cut vertices of~$G$, and~$\ell$ the number of blocks in $G$. Then,
\[
\gamma_c(G) = 
\begin{cases}
  1 	& \text{for}\ \ell = 1,	\\
  |S|   & \text{for}\ \ell \geq 2.
\end{cases}
\]
\end{theorem}
Combining the two previous theorems, we obtain the following.
\begin{theorem}
\label{thm_rc_upper_bound}
Let $G$ be a connected block graph with at least two blocks and $\delta(G) \geq~2$. Then $\rc(G) \leq |S|+2$, where $S$ is the set of cut vertices of $G$. Moreover, this bound is tight.
\end{theorem}
Finally, as the number of cut vertices can be determined in linear time, we remark that the upper bound can also be computed in linear time. 

Before proceeding, let us state the following simple but useful lemma.
For this, we recall that a \emph{peripheral vertex} is a vertex of maximum eccentricity, that is, a vertex that is a starting point for some diametral path.
\begin{lemma}
\label{thm_sep_touching_two_pvertices}
Let $G$ be a block graph with at least 3 blocks, and let $x$ and $y$ be two peripheral vertices in distinct blocks. If $G$ has a cut vertex $s$ adjacent to $x$ and $y$, then $\rc(G) > \diam(G)$.
\end{lemma}
\begin{proof}
For the sake of contradiction, assume that $\rc(G) = \diam(G)$. 
Let $B_x$ and $B_y$ be the two distinct blocks $x$ and $y$ are in, respectively. 
Choose a vertex $z \in B_z$ such that $d(x,z) = \diam(G)$, where $B_z$ is a block different from $B_x$ and $B_y$. 
Let $P_{xz}$ be the (unique) shortest $x$-$z$ path.
Since $\rc(G) = \diam(G)$, it must be the case that each edge in $E(P_{xz})$ receives a distinct color in any valid rainbow coloring using $\diam(G)$ colors.
Since $x$ and $z$ are in distinct blocks, and because $s$ is a cut vertex, it is clear that $(x,s) \in E(P_{xz})$.
Without loss, suppose the edge $(x,s)$ was colored with color $c_1$. 
Then consider each uncolored edge incident to $s$ in $B_y$. 
Notice we must color each such edge with color $c_1$, for otherwise $G$ would not be rainbow-connected.
But now the (unique) shortest path $P_{xy}$ repeats the color $c_1$, and thus $x$ and $y$ are not rainbow-connected.
It follows that $\rc(G) > \diam(G)$.~$\square$
\end{proof}

\noindent Figure~\ref{fig:tight_approx_example} (a) illustrates the previous claim: the block graph $G$ has two peripheral vertices adjacent to a cut vertex $s$. Both the edges $(x,s)$ and $(y,s)$ would have to receive the same color in a rainbow coloring of $G$ using $\diam(G)$ colors, but then there is no way to rainbow-connect $x$ and $y$ without introducing new colors. (Here, Figure~\ref{fig:tight_approx_example} (a) also shows the bound from Theorem~\ref{thm_rc_upper_bound} is tight).

\begin{figure}[t]
  \begin{center}
    \subfigure[]{\includegraphics[scale=1,keepaspectratio]{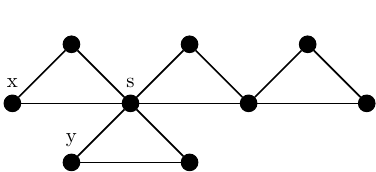}}
    \hfil
    \subfigure[]{\includegraphics[scale=1,keepaspectratio]{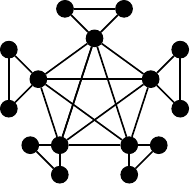}}
    \caption{\textbf{(a)} A block graph $G$ with a cut vertex $s$ adjacent to two peripheral vertices $x$ and $y$ in distinct peripheral blocks. \textbf{(b)} A $K_n$ with $n$ triangles glued to it for $n = 5$.}
    \label{fig:tight_approx_example}
  \end{center}
\end{figure}

We will then characterize the bridgeless block graphs having a rainbow connection number 2, 3, or~4. The following also determines exactly the rainbow connection number of the \emph{windmill graph} $K^{(m)}_n$ ($n > 3$), which consists of $m$ copies of $K_n$ with one vertex in common.

\begin{theorem}
\label{thm_rc_cases}
Let $G$ be a bridgeless block graph. Deciding whether $\rc(G) = k$ is in $\P$ for $k \in \{1,2,3,4\}$.
\end{theorem}
\begin{proof}
As $k \leq 4$, it is enough to consider bridgeless block graphs with diameter $d = \diam(G) \leq 4$. In what follows, we show how such graphs are efficiently and optimally colored.
\begin{itemize}
\item Case $d = 1$. Trivial as $G$ is complete.

\item Case $d = 2$. If $G$ has exactly 2 blocks, it is easy to see that $\rc(G) = 2$. Moreover, if the graph has $\rc(G) = 2$, it must have exactly 2 blocks. Suppose this is was not the case, i.e., $G$ has at least 3 blocks and $\rc(G) = 2$. By an argument similar to Lemma~\ref{thm_sep_touching_two_pvertices}, this leads to a contradiction. Thus, $\rc(G) = 2$ if and only if $G$ has exactly 2 blocks. When $G$ consists of 3 or more blocks, we will show that $\rc(G) = 3$. Let $\mathcal{B}$ be the set of all blocks of $G$, and let $a$ be the unique central vertex of $G$. For each $B \in \mathcal{B}$, color one edge incident to $a$ with color $c_1$, and every other incident edge with color $c_2$. Finally, color every uncolored edge of $G$ with color $c_3$. 
To see $G$ is rainbow-connected, observe there is a rainbow path from any vertex to the central vertex $a$ avoiding a particular color in $\{c_1,c_2,c_3\}$.

\item Case $d = 3$. 
The graph $G$ consists of a unique central clique, and at least 2 other blocks.
If $G$ has 3 blocks, then $\rc(G) = \src(G) = 3$.
If $G$ has 4 blocks, there are two cases: either $G$ has a cut vertex adjacent to two peripheral vertices in distinct blocks (then $\rc(G) \geq 4$ by Lemma~\ref{thm_sep_touching_two_pvertices}) or it does not (then $\rc(G) = \src(G) = 3$).
Otherwise, $G$ has at least 5 blocks.
Now, if $G$ has exactly two cut vertices, then by an argument similar to Lemma~\ref{thm_sep_touching_two_pvertices}, we have that $\rc(G) \geq 4$ as there is a cut vertex that is contained in more than two blocks.
We will then color every block that is not the central clique with 3 colors exactly as in the case $d = 2$, and color every edge of the central clique with a fresh distinct color proving that $\rc(G) = 4$.
Finally, suppose $G$ has at least 5 blocks, more than two cut vertices, but every cut vertex is contained in exactly two blocks.
Then by an argument similar to Lemma~\ref{lem:src_lower_bound}, we have that $\rc(G) \geq 4$, and the described coloring proves $\rc(G) = 4$.

\item Case $d = 4$. Let us call the set of blocks which contain the central vertex $a$ the \emph{core} of the graph $G$. The set of blocks not in the core is the \emph{outer layer}. First, suppose the core contains exactly 2 blocks, and the outer layer at most 4 blocks. Furthermore, suppose the condition of Lemma~\ref{thm_sep_touching_two_pvertices} does not hold (otherwise we would have $\rc(G) > 4$ immediately). When the outer layer contains 2 or 3 blocks, we have that $\rc(G) = \src(G) = 4$. Suppose the outer layer contains exactly 4 blocks. First, consider the case where a core block is adjacent to 3 blocks in the outer layer. Because the condition of Lemma~\ref{thm_sep_touching_two_pvertices} does not hold, it must be the case that at least one of the core blocks is not a $K_3$. Clearly, every two vertices $x$ and $y$, such that $d(x,y) = \diam(G)$, have to be connected by a rainbow shortest path. By an argument similar to Lemma~\ref{lem:src_lower_bound}, we have that $\rc(G) > 4$. Otherwise, when a core block is not adjacent to 3 blocks in the outer layer, $\rc(G) = \src(G) = 4$. Now suppose the outer layer has at least 5 blocks. As above, by an argument similar to Lemma~\ref{lem:src_lower_bound}, we have that $\rc(G) > 4$. Finally, suppose the core has 3 or more blocks. We argue that in this case, $\rc(G) = 4$ if and only if the outer layer contains exactly 2 blocks. For the sake of contradiction, suppose $\rc(G) = 4$, and that the outer layer has 3 or more blocks. If the condition of Lemma~\ref{thm_sep_touching_two_pvertices} holds, we have an immediate contradiction. Otherwise, by an argument similar to Lemma~\ref{thm_sep_touching_two_pvertices}, we arrive at a contradiction. When the outer layer contains exactly 2 blocks, we will show $\rc(G) = 4$. Let $B_1$ and $B_2$ be the blocks in the outer layer. We color every edge of $B_1$ with the color $c_1$, and every edge of $B_2$ with the color $c_4$. Then color $(b_1,a)$ with $c_2$, and $(a,b_2)$ with $c_3$, where $a$ is the central vertex of $G$, and $b_1$ and $b_2$ are the cut vertices in $B_1$ and $B_2$, respectively. For every block $B_i$ in the core, let $Q_i$ denote the set of edges in $B_i$ incident to $a$. Color the uncolored edges of $Q_i$ with either $c_2$ or $c_3$, such that both colors appear at least once in $Q_i$. Then, color every uncolored edge of the block that contains both $a$ and $b_2$ with the color $c_1$. Every other uncolored edge of $G$ receives the color~$c_4$. 
We can now verify $G$ is indeed rainbow-connected under the given coloring.~$\square$
\end{itemize}
\end{proof}
It appears plausible but tedious that one could extend the theorem for larger values of $k$ as well. Thus, it is perhaps the case that deciding whether $\rc(G) = k$ is solvable in polynomial time for any fixed $k$ where~$G$ is a block graph. However, we conjecture the following.
\begin{conjecture}
Given a block graph $G$, it is $\NP$-hard to rainbow color $G$ optimally.
\end{conjecture}
Put differently, the conjecture says the decision problem is $\NP$-complete when the number of colors $k$ is not fixed but part of the input.

Given that the strong rainbow connection number of a block graph $G$ can be efficiently computed, it is interesting to ask when $\rc(G) = \src(G)$, or if the difference between $\src(G)$ and $\rc(G)$ would always be small. Because $\diam(G) \leq \rc(G)$ for any connected graph $G$, the following is an easy observation.
\begin{corollary}
Let $G$ be a block graph, and let $k$ be the number of blocks containing less than 3 cut vertices. If $k = \diam(G)$, then $\rc(G) = \src(G)$.
\end{corollary}
However, the difference between $\src(G)$ and $\rc(G)$ can be made arbitrarily large: attach $n$ triangles to a $K_n$, one to each vertex of the $K_n$ (see Figure~\ref{fig:tight_approx_example} (b) for an illustration). As $n$ increases, the rainbow connection number remains 4 by Theorem~\ref{thm_rc_cases}, while the strong rainbow connection number increases by Theorem~\ref{thm_src_equals_k}. This example also shows the difference between the upper bound of Theorem~\ref{thm_rc_upper_bound} and $\rc(G)$ can be arbitrarily large.

\section{Concluding remarks}
We studied the complexity of computing the rainbow and strong rainbow connection numbers of subclasses of chordal graphs, namely split graphs and block graphs. In particular, Theorem~\ref{thm_src_chordal_hardness} shows the strong rainbow connection number is significantly harder to approximate than the rainbow connection number, even on very restricted graph classes. Indeed, the result should be contrasted with the fact that any split graph can be rainbow colored in linear time using at most one color more than the optimum~\citep{Chandran2012}.

We believe our results for rainbow and strong rainbow coloring block graphs can serve as a starting point for an even more systematic study of strong rainbow coloring more general graph classes --- a topic which has received quite little attention despite the interest. In fact, the investigation of the strong rainbow connection number has been deemed ``much harder than that of rainbow connection number''~\cite{Li2012} (see also~\cite{Li2012} for more discussion). Given this observation, it is meaningful to consider the strong rainbow connection number of a most general restricted graph class (e.g., block graphs) for which the computation of the number is not known to be $\NP$-complete.

Finally, to avoid confusion, we note that similar problems have been considered in e.g., \citep{Uchizawa2013,Lauri15}: given an edge-colored graph $G$, decide whether $G$ is (strongly) rainbow-connected. We stress that known hardness results for these problems do not imply hardness results for \emph{finding} rainbow colorings. Indeed, the problems are strictly different.

%\acknowledgements
%\label{sec:ack}
%At the end of the manuscript, right before the bibliography you might
%want to place an acknowledgment. This can be easily done by using the 
%command \verb!\acknowledgements! as you can see here.

%\nocite{*}
\bibliographystyle{abbrvnat}

% use the following instead if you encounter problems 
%\bibliographystyle{alpha}
%\bibliographystyle{abbrv}

\bibliography{bibliography}
\label{sec:biblio}

\end{document}